\documentclass[final]{siamltex}

\newtheorem{claim}[theorem]{Claim}
\newtheorem{notation}{Notation}

\newtheorem{remark}[theorem]{Remark}

\newcommand{\lab}{\left | }
\newcommand{\rab}{\right | }

\newcommand{\lcub}{\left \{ }
\newcommand{\rcub}{\right \} }

\newcommand{\te}{\rightarrow}

\newcommand{\pa}{\partial}

\title{On the Steady States of Weakly Reversible Chemical Reaction Networks}

\author{Jian Deng\thanks{CEMA, Central University of Finance and Economics, Beijing, P.R.China,
100085({\tt jdeng@fudan.edu.cn}).}\and Martin Feinberg \thanks{Department of Mathematics, Ohio State University, Columbus, Ohio}
        \and Chris Jones\thanks{Department of Mathematics,
University of North Carolina at Chapel Hill, Chapel Hill, NC 27599} \and Adrian Nachman\thanks{Department of Mathematics, University of
Toronto, Canada M5S 3G3}}

\begin{document}

\maketitle

\begin{abstract}
A natural condition on the structure of the underlying chemical reaction network, namely weak reversibility, is shown to guarantee the existence of an equilibrium (steady state) in each positive stoichiometric compatibility class for the associated mass-action system. Furthermore, an index formula is given for the set of equilibria in a given stoichiometric compatibility class.
\end{abstract}

\begin{keywords}
chemical reaction network, weak
reversibility, equilibria.
\end{keywords}

\begin{AMS}
34C10, 37N99
\end{AMS}

\pagestyle{myheadings}
\thispagestyle{plain}
\markboth{J. Deng et al.}{Steady state of chemical reaction networks}

\section{Introduction}

 \hspace{0.5cm} The goal of chemical reaction network theory is to
formulate conditions under which the dynamical fate of composition
trajectories can be ascertained, even in complex reactions.
Chemical reaction systems considered in practical applications are
often very complicated. In principle, the number of species
involved can be arbitrarily high. In practice, a full system can
involve tens or hundreds of species. Although models with a
handful of species are usually used, these systems are still
``high'' dimensional from the perspective of dynamical system.
Moreover, the reactions between the complexes is often only known
approximately, since the determination of rate constants is quite
difficult, if not impossible. A natural approach is therefore to
attempt to give a qualitative description of the behavior of the
reaction system.

The modeling of chemical reactions can be achieved via the
mass-action assumption. A typical chemical reaction network (CRN) obeying mass-action kinetics
consists of three components: $\lcub \mathcal{S,Y}, K=(k_{ij})  \rcub$, where

\begin{enumerate}
\item $\mathcal{S} = \lcub S_1, S_2, ..., S_n \rcub $ is the set
of species involved in the chemical reaction.

\item $\mathcal{Y} = \lcub Y_1, Y_2, ..., Y_m \rcub $ is the set
of complexes in the given CRN, with each $Y_i \in R^n,
i=1,2,\ldots,m$.

\item If we have a reaction $Y_i \te Y_j$, then $k_{ij} = k_{Y_i
\te Y_j} > 0$ is the \emph{rate constant} for the reaction $Y_i
\te Y_j$, assuming mass-action kinetics. If there is no reaction
from $Y_i$ to $Y_j$, then we define $k_{ij} = k_{Y_i \te Y_j} =
0$.

\end{enumerate}

In this framework, the equation governing the evolution of
chemical concentrations is given by:

\begin{equation}\label{dyn}
\dot x = f(x) = \sum_{i,j = 1}^m
k_{ij}x^{Y_i}(Y_j - Y_i)
\end{equation}

\noindent where $x \in R^n_+ = \lcub x = ( x_1, \ldots, x_n ) :
x_i > 0, i = 1,2,\ldots,n \rcub$, and each $x_i$ is the
concentration of species $S_i$ in the chemical reaction network. \bigskip

The system we obtain is an array of ordinary
differential equations with polynomial vector fields determining
the behavior of the concentration of each species in the reactor.
This ODE could be quite complicated: the dimension of the equation
could be dauntingly high, while the number of parameters (rate
constants, for example) could be equally large. It is therefore
quite surprising, but nevertheless true, that many system exhibit
relatively simple dynamics. In fact, based on considerable
experience, chemical experimentalists have developed an
``intuition'' that a ``normal'' chemical experiment will quickly
lead to some steady state.

So one of the primary concerns of chemical reaction network
theory, up to now, has been to determine the capacity of a given
reaction network system for multiple steady states. This theme
emerged from the early work of chemical reaction network theory
(\cite{FH1}, \cite{HJ}), where chemical reaction networks were
classified according to their deficiency. Horn, Jackson and
Feinberg give a quite complete picture of the behavior of the
deficiency-zero and deficiency-one mass-action systems (\cite{F3},
\cite{F4}). Introductory material for chemical reaction network
theory is given in (\cite{F1}, \cite{F2}), and we will follow the
notation and definitions in these papers as much as possible.

The capacity of a reaction network for multiple states depends
heavily on the underlying algebraic structure of the reaction
system, given by the ``react to'' relation between the complexes.
Two complexes are linked together if there is a reaction between
them, and we say two complexes are in the same \emph{linkage
class} if we can find a finite sequence of reactions between them. It is
not hard to check that the linear subspace spanned by all the
reaction vectors $Y_j - Y_i$, for each reaction $Y_i \te Y_j$, is
invariant under the flow induced by (\ref{dyn}). It is called the
\emph{stoichiometric subspace}. Each linear manifold parallel to
it is called a \emph{stoichiometric compatibility
class}(\emph{SCC}), which is also invariant under the flow.

The ``react to'' relation between the complexes also gives us
various notion of reversibility, as alluded to above. If each
reaction $A \te B$ in the chemical reaction network is accompanied by the reverse reaction $B \te
A$, then this reaction network is called \emph{reversible}. If for
each reaction $A \te B$ we have a  reaction chain  from $B$ to
$A$, i.e., we can find a finite sequence of reactions $B \te C_1, C_1 \te
C_2, \ldots, C_{n-1} \te C_n, C_n \te A$, then this reaction
system is called \emph{weakly reversible}. We would consider
solely weakly reversible chemical reaction networks in this paper.

In the early eighties, Nachman obtained the existence of steady
states in each stoichiometric compatibility class if there is only
one linkage class in the chemical reaction network (\cite{N}).
Using this result, Feinberg showed the existence of steady states
in each SCC (personal communications) if the chemical reaction
network is weakly reversible and the deficiency is equal to the
summation of the deficiencies of each linkage class. So the next natural
question is about the existence of steady states for general
weakly reversible systems.



\bigskip

We are interested in the following two problems:

\begin{enumerate}

\item Does there exist a positive steady state for (\ref{dyn})?

\item Does there exist a steady state in each nonempty, positive
stoichiometric compatibility class for (\ref{dyn})?

\end{enumerate}

The main result in this paper gives a affirmative answer to the two
questions above, under a natural condition.

\begin{theorem} For each weakly reversible chemical reaction network obeying mass-action kinetics, the flow of (\ref{dyn}) has finitely many (at least one) steady states in each nonempty, positive stoichiometric compatibility class.
\end{theorem}

 Our result shows that the experimental chemists' ``intuition'' is well grounded, at least
for weakly reversible and mass-action chemical reaction networks.
It shows that for such systems there will always be some steady
state in each stoichiometric compatibility class(SCC), and
although the numbers of steady state may differ for different SCC,
the summation of the indices for the steady states remains the
same (see Corollary $5.2$). Moreover, this does not depend on the specific
values of the positive rate constants, as long as the system
remains weakly reversible. This would be quite appealing to experimental
chemists since the rate constants are difficult to measure.

\noindent \begin{remark} \rm The component of the complexes $Y_i$
is called the stoichiometric coefficients, and for typical chemical
reactions, it should be a nonnegative integer. But our result
holds when the stoichiometric coefficients are real, thus for
example we allow the following reaction:

  $$ 0 \te 2A + B, \hspace{.5in} 2.5B + C \te 0, \hspace{.5in} -2C + D = E $$

\noindent as long as the chemical reaction is weakly reversible.
Therefore the case where there is ``source'' and ``sink'' for the
chemical reaction system can still be described by our result.

\end{remark}

\bigskip



\noindent \emph{Plan of the paper.}  In section 2, we define two matrices $C, \mathbf{R}$ that
determine the reaction network structure, discuss their properties
and reformulate the existence problem as that of an intersection
problem of two hypersurfaces. In section 3, we restrict ourselves
to the case of one linkage class. A vector-valued function $G(z)$ is defined, and a priori estimate of $(G(z),z)$ is
given and used to solve the intersection problem in a special
case. In section 4, we discuss the case of $l$ linkage classes and
construct a bounded, convex set so that the intersection problem
is reduced to proving that a corresponding vector field is pointing
in at each point of its boundary. Thus we give the proof of the
existence of steady states in the whole $R^n_+$ space. In section
5, we utilize the remark given at the end of section 4 to prove
the existence of a steady state in each stoichiometric
compatibility class, and then we use a homotopy argument to give
the proof of the index formula.


\section{Reaction network structure and reformulation of the equation}
\hspace{0.5cm} It has long been known that the existence of steady
states depends heavily on the algebraic structure of the
chemical reaction network. For example, when the reaction network is reversible, forest-like
and the deficiency of the network is 0, one can show the existence of steady states in $R^n_+$ (the reader is referred to (\cite{J}) concerning the precise exposition). But it is not obvious
how to extend the result to the weakly reversible reaction networks of arbitrary
deficiency.

The information about a weakly reversible reaction network can essentially
be decomposed into two parts: one about the configuration of the
complexes in the $R^n$ space, offered by the set $\mathcal{Y} =
\lcub Y_1, Y_2, \ldots, Y_m \rcub$; the other one about the
reaction information between the complexes, given by the
rate constants $k_{ij}, i, j = 1, 2, \ldots, m$. We
give the following

\begin{definition}  For each weakly reversible chemical reaction network $\lcub \mathcal{S,Y}, K  \rcub$, the \emph{configuration matrix}  $C$ is given by
$$ C = \left( \begin{array}{c}
Y_1 \\ Y_2 \\ \vdots \\ Y_m \end{array} \right)_{m \times n} = \left( \begin{array}{cccc} y_{11} & y_{12} & \ldots & y_{1n} \\
   y_{21} & y_{22} & \ldots & y_{2n} \\
   \vdots & \vdots & \vdots & \vdots \\
   y_{m1} & y_{m2} & \ldots & y_{mn}
   \end{array} \right), $$
\noindent while the \emph{relation matrix} ${\mathbf R}$ is given by \\

\hspace{2.5cm} ${\mathbf R} = (r_{ij})_{m \times m}$, with $r_{ij} = \left \{ \begin{array}{ll} k_{ij}, & if \quad i \ne j. \\
                                                      -\sum_{s \ne i} k_{is}, & if \quad i = j.
                                    \end{array} \right. $
\end{definition}

Therefore the reaction network structure is completely determined
by the two matrices $\lcub C, \mathbf{R} \rcub$. Notice that
\emph{the property of weak reversibility is uniquely determined by
the relation matrix ${\mathbf R}$}. They will both play a role in
the existence of the positive equilibrium, but it will turn
out that the specific structure of $\mathbf{R}$ will be crucial
in the following development.

System (\ref{dyn}) can be rewritten in terms of $\lcub C, \mathbf{R}
\rcub$:

\begin{equation}\label{refo}
\left \{
\begin{array}{ll}
\dot x & = \left( \begin{array}{cccc} x^{Y_1} & x^{Y_2} & \ldots &
x^{Y_m} \end{array} \right)_{1 \times m} {\mathbf R}_{m \times m}
\left( \begin{array}{c}
Y_1 \\ Y_2 \\ \vdots \\ Y_m \end{array} \right)_{m \times n} \\
& = \left( \begin{array}{cccc} e^{\ln x \cdot Y_1} & e^{\ln x
\cdot Y_2} & \ldots & e^{\ln x \cdot Y_m} \end{array} \right)_{1
\times m} \cdot {\mathbf R} \cdot C
\end{array}\right .
\end{equation}

\noindent where $\lcub C, \mathbf{R} \rcub$ are given as above. \bigskip

We list some basic properties of the relation matrix
$\mathbf{R}$.

\begin{enumerate}

\item ${\mathbf R} = (r_{ij})_{m \times m}$ satisfies:

\begin{itemize}

\item $ \sum_{j =1}^m r_{ij} = 0$, for $\hspace{0.2cm}$$i = 1, 2,
\ldots, m$

\item $\left \{
       \begin{array}{ll}
       r_{ij} \ge 0, & for \quad i \ne j \\
       r_{ii} < 0, & for \quad i = 1, 2, \ldots, m
       \end{array} \right.$

\end{itemize}

\item Every two complexes in the same linkage class are linked by
``reaction chains''. Mathematically, this means for each pair
$Y_i, Y_j$ in a certain linkage class, there exist $Y_{i_1},
Y_{i_2}, \ldots, Y_{i_s}$ such that $Y_{i_1} = Y_i, Y_{i_s} = Y_j$
and $r_{i_{k-1}i_k}> 0 $ for $k = 2, 3, \ldots, s$.

\item If the CRN has $l$ linkage classes, $l > 1$, then by relabeling $Y_1,
Y_2, \ldots, Y_m$ we can write ${\mathbf R}$ as the following
diagonal block form

\begin{equation}\label{block}
{\mathbf R} = \left( \begin{array}{cccc}
                  {\mathbf R}_1    \\
                   & {\mathbf R}_2  \\
                   & &  \ddots &     \\
                   & & & {\mathbf R}_l
                 \end{array} \right)
\end{equation}

\noindent with each ${\mathbf R}_i$ being the relation matrix for
a linkage class.

\end{enumerate}

The only thing in item $1$ which needs explanation is $r_{ii} <
0$. If $r_{ii} = 0$ for some $i$, then $r_{ij} = 0$ for $j \ne i$.
Thus the complex $Y_i$ does not react to any other complexes. By
weak reversibility of the chemical reaction network, no other
complex will react to $Y_i$ either. Thus $Y_i$ can safely be
discarded and we only need to consider the subnetwork without $Y_i$.
Therefore we will always assume that each complex $Y_i$ appears at
least once in the chemical reaction network.

\begin{remark} \rm It is interesting to observe that the classical Perron-Frobenius theory for matrices with positive entries can be used to get information about the eigenvalue distribution of the relation matrix ${\mathbf R}$. For example, we can prove that ${\mathbf R}$ has only nonpositive eigenvalues, and the algebraic multiplicity of the zero eigenvalue of ${\mathbf R}$ is exactly the number of linkage classes within the CRN. Also, the eigenvector corresponding to the zero eigenvalue is nonnegative. This spectral information seems to have decisive influence on the stability property of the phase portrait of the flow. The interested reader is referred to (\cite{G}).
\end{remark}

\begin{notation} For each CRN we define the \emph{norm}
of its relation matrix $\mathbf{R}$ to be $\lab \mathbf{R} \rab =
\max_{i =1,2,\ldots,m}\lab r_{ii} \rab$. Also, we define
$\tau(\mathbf{R}) = \min_{i \ne j, i,j \in \lcub 1,2,\ldots,m
\rcub}\lcub r_{ij} : r_{ij} > 0 \rcub$. \end{notation}

\bigskip

Our next task in this section is to transfer the existence problem
to one of intersection of two hypersurfaces. The idea is to see
that $x \in R_+^n $ is an equilibrium point for system (\ref{refo}) if and
only if

$$ \left( \begin{array}{cccc} e^{\ln x \cdot Y_1} & e^{\ln x \cdot Y_2} & \ldots & e^{\ln x \cdot Y_m} \end{array} \right)_{1 \times m} \cdot {\mathbf R} \in \lcub z \in R^m : z \cdot C = 0 \rcub ,$$

\noindent So it is natural to consider the intersection

\begin{equation}\label{inter}
 \lcub \left( \begin{array}{cccc} e^{\ln x \cdot Y_1} & e^{\ln x \cdot Y_2} & \ldots & e^{\ln x \cdot Y_m} \end{array} \right) \cdot {\mathbf R}: x \in R_+^n \rcub \cap \lcub z \in R^m : z \cdot C = 0 \rcub
\end{equation}

\noindent if we want to find positive steady states for (\ref{refo}).

\begin{notation} $ N = \lcub z \in R^m : z \cdot C = 0
\rcub $, $K = \lcub z \in R^m :z = x \cdot C^t \hspace{0.1cm} for
\hspace{0.1cm} some \hspace{0.1cm} x \in R^n \rcub $.\end{notation}

Note that $K = N^{\bot},  K^{\bot} = N$ where $N^{\bot}$ and $K^{\bot}$ mean the orthogonal
subspaces of $N$ and $K$ in $R^m$ respectively.

\begin{definition} \rm We define the \emph{relation function} $G(z): R^m \te R^m$ as

\begin{equation}
 G(z) \stackrel{def}= e^z \cdot {\mathbf R} = (
e^{z_1} e^{z_2} \ldots z^{z_m}) {\mathbf R}.
\end{equation}
\end{definition}

Then we have the following

\begin{proposition} There exists a positive equilibrium point for (\ref{refo}) in $R^n_+$ if and only if

\begin{equation}\label{intersec}
 G(K) \cap K^{\bot} \ne \emptyset.
\end{equation}

\end{proposition}

\begin{proof}: It suffices to notice that

\noindent $$ G(K) =  \lcub \left( \begin{array}{cccc} e^{\ln x
\cdot Y_1} & e^{\ln x \cdot Y_2} & \ldots & e^{\ln x \cdot Y_m}
\end{array} \right)_{1 \times m} \cdot {\mathbf R}: x \in R_+^n
\rcub $$

\noindent and

\noindent $$ K^{\bot} = N = \lcub z \in R^m : z \cdot C = 0 \rcub
$$

\noindent then from equation (\ref{inter}) and the preceding discussion, we
obtain the result.\end{proof}

\begin{remark} The idea of using the intersection of
two hypersurfaces to prove an existence result can be traced to
Felix E. Browder in (\cite{B}). Although it turns out that the
relation function $G(z):R^m \te R^m$ is not monotone,
the fact that we are dealing with a \emph{finite} dimensional ODE
still pulls us through.
\end{remark}

\section{The case of one linkage class}





\hspace{0.5cm} In this section we will solve the intersection
problem, assuming that the CRN has only one linkage class. The general strategy is as follows: to prove that $G(K) \cap K^{\bot} \ne \emptyset$, it is equivalent to showing that

$$ 0 \in \Pi_K \circ G(K), \hspace{7cm} (*) $$

\noindent where $\Pi_K : R^m \te K$ is the standard projection
operator. One sufficient condition for (*) to hold is that there exists a
large ball $B(r) = \lcub z \in K: \lab z \rab < r \rcub$ in $K$
such that

$$ (\Pi_K \circ G(z), z) < 0 $$

\noindent for $z \in \partial B(r) = \lcub z \in K: \lab z \rab =
r \rcub$. Then by the Brouwer Fixed Point theorem we obtain (*).

So basically, we need to estimate

$$ (\Pi_K \circ G(z), z) = (G(z), z) $$

\noindent for $z \in K$ and try to show that we can find a ball
$B(r)$ such that

$$ (G(z), z) < 0 $$

\noindent for $z \in \pa B(r)$. But here we encounter some technical difficulties: one quickly finds
that $(G(z), z)$ will be identically zero if $z$ lies in the $1-d$
subspace $D = \lcub z = a \cdot (1,1,\ldots,1) \in R^m : a \in R
\rcub $. Thus we will have to avoid this degenerate subspace $D$.

\begin{definition}We set ${\underline 1}_m = (1,1, \ldots, 1)_{1
\times m}, P_m = \lcub z \in R^m : {\underline 1}_m \cdot z = 0\rcub $, and we define the \emph{norm} on $R^m$ as $\lab z \rab_2 = \sum_{i=1}^m {z_i}^2$. If $z \in P_m$, we define another equivalent norm $ \lab z \rab = \max_{i = 1,2,\ldots, m} z_i$.
\end{definition}

We start from the following

\begin{lemma}Let $s \ge 1$ be fixed. Then

\begin{equation}\label{est}
 \max_{-L \le y_i \le 0, i = 1,2,\ldots, s}
y_1 + e^{y_1}y_2 + \ldots + e^{y_{s-1}}y_s + e^{y_s}(-L) \te
-\infty
\end{equation}

\noindent as $ L \te +\infty$.

\end{lemma}

\begin{proof} We will show (\ref{est}) by induction on $s$.

\noindent Step 1: If s =1, then

$$ \max_{-L \le y_1 \le 0} y_1 + e^{y_1}(-L) = \max \lcub -L, -L + e^{-L}(-L), -1-\ln L \rcub \te -\infty $$

\noindent as $L \te +\infty$.

\noindent Thus the lemma is true for $s =1$.

\noindent Step 2: Now suppose that the lemma is true for $s \le
k$. Define

$$ A_{k,L} = \max_{-L \le y_i \le 0, i =1,2,\ldots,k} \lcub y_1 + e^{y_1}y_2 + \ldots + e^{y_k}(-L) \rcub $$

\noindent then by assumption we know that $A_{k,L} \te -\infty$ as
$L \te +\infty$.

\noindent For $s = k+1$, let $ -L \le y_1, y_2, \ldots, y_k \le 0$
be fixed, then

$$ \Psi(y_{k+1};y_1,y_2,\ldots,y_k) = y_1 + e^{y_1}y_2 + \ldots + e^{y_k}y_{k+1} + e^{y_{k+1}}(-L) $$

\noindent will achieve its maximum at $y_{k+1} = 0, -L$ or $ y_k -
\ln L$. So we have

\begin{equation}
 A_{k+1,L} = \max_{-L \le y_i \le 0, i
=1,2,\ldots, k+1}\Psi(y_{k+1};y_1,y_2,\ldots,y_k) \le \max \lcub
-L, A_{k,L}, A_{k,\ln L} \rcub,
\end{equation}

\noindent therefore $A_{k+1,L} \te -\infty $ as $ L \te +\infty$.
Thus the lemma is true for $s = k+1$.

\noindent Step 3: Combining Steps 1 and 2 we see that the lemma is
true for all $ s \ge 1$.
\end{proof}

\bigskip

The fundamental lemma is the following

\begin{lemma} If the CRN has exactly one linkage class, then there exists a function $c(r)$ on $R^1$ with $c(r) \te +\infty$ as $r \te +\infty$ such that

\begin{equation}\label{decay}
(G(z), z) \le - c(\lab z \rab)e^{\lab z
\rab}
\end{equation}

\noindent for all $z \in P_m$.

\end{lemma}

\begin{proof} For all $z \in P_m$, we have

\begin{equation}
\left \{
\begin{array}{ll}
\frac{(G(z),z)}{e^{\lab z \rab}} & = \frac{e^z{\mathbf R}z^t}{e^{\lab z \rab}} \\
 = e^{z-\lab z \rab \underline 1_m}{\mathbf R}z^t & = e^{z-\lab z \rab \underline 1_m}{\mathbf R}( z^t - \lab z \rab \underline 1_m^t) = e^y{\mathbf R}y^t
\end{array} \right .
\end{equation}

\noindent where $y = z - \lab z \rab \underline 1_m$.

\begin{notation} $B_L = \lcub y \in R^m : -L \le y_i \le
0, \forall i =1,2,\ldots,m; \min_{i}y_i = -L, \max_{i} y_i = 0
\rcub $. \end{notation}

Then lemma 3.3 is equivalent to the following

\bigskip

\noindent {\bf Claim:} Define $ F: B_L \te R$ as $F(y) = e^y
{\mathbf R}y^t$, $y \in B_L$, then

$$ C_L = \max_{y \in B_L}F(y) \te -\infty $$ as $L \te +\infty$.

\noindent Proof of Claim: Given $y \in B_L$, by relabeling $Y_1,
Y_2, \ldots, Y_m$ we can write $y$ as

$$ y = (\underline 0_{t_1}, \underline {\hat y}_{t_2}, -L \underline 1_{t_3}) $$

\noindent where $ \left \{
      \begin{array}{ll}
       -L < \hat y_i < 0, \forall 1 \le i \le t_2 \\
       t_1 + t_2 + t_3 = m, t_1 \ge 1, t_3 \ge 1
      \end{array} \right. $.

So we have

\begin{equation}
\left .
\begin{array}{ll}
F(y) & = e^y{\mathbf R}y^t = \left( \begin{array}{ccc}
e^{\underline 0_{t_1}} & e^{\underline {\hat y}_{t_2}} & e^{-L
\underline 1_{t_3}} \end{array} \right) \left( \begin{array}{ccc}
                  {\mathbf R}_{11} & {\mathbf R}_{12} & {\mathbf R}_{13} \\
                  {\mathbf R}_{21} & {\mathbf R}_{22} & {\mathbf R}_{23}  \\
                  {\mathbf R}_{31} & {\mathbf R}_{32} & {\mathbf R}_{33}
                 \end{array} \right)  \left( \begin{array}{c} \underline 0_{t_1}^t \\ \underline {\hat y}_{t_2}^t \\ -L \underline 1_{t_3}^t \end{array} \right) \\
\le & \underbrace{e^{\underline 0}{\mathbf R}_{12}\underline {\hat
y}^t + e^{\hat y}{\mathbf R}_{23}(-L \underline 1^t)}_{(1)} +
\underbrace{ e^{\underline 0}{\mathbf R}_{13}(-L \underline
1^t)}_{(2)} + \underbrace{ e^{-L \underline 1}{\mathbf R}_{33}(-L
\underline 1^t) +  e^{\underline {\hat y}}{\mathbf
R}_{22}\underline {\hat y}^t }_{(3)}
\end{array} \right .
\end{equation}

Estimate of term (3): we have

\begin{equation}
\left \{
\begin{array}{ll}
term (3) = & e^{-L}(-L)\underline 1_{t_3} {\mathbf R}_{33} \underline 1_{t_3}^t + e^{\underline {\hat y}}{\mathbf R}_{22}\underline {\hat y}^t \le t_3 \lab {\mathbf R} \rab L e^{-L} + t_2 \lab {\mathbf R} \rab (\max_{-\infty < s \le 0}\lab e^s s \rab) \\
& \le (t_2+t_3) \lab {\mathbf R} \rab e^{-1} \le \frac{m}{e} \lab
{\mathbf R} \rab
\end{array} \right .
\end{equation}

To estimate terms (1) and (2) we need to consider two cases:

\begin{itemize}

\item ${\mathbf R}_{13} \ne 0$, by which we mean there is at least
one element $ a_{ij} \in {\mathbf R}_{13}$ which is larger than
$0$, then

$$ term (2) \le (-L)a_{ij} \le -L \cdot \tau({\mathbf R}). $$

Notice that term (1) will always be nonpositive since each element
in ${\mathbf R}_{12}, {\mathbf R}_{23}$ is nonnegative, and each
element in $\underline {\hat y}_{t_2}$ and $-L \underline 1_{t_3}$
is strictly negative.

So we have

\begin{equation}\label{key_est1}
 term (1) + (2) + (3) \le \frac{m}{e} \lab
{\mathbf R} \rab - L \cdot \tau({\mathbf R})
\end{equation}

\item If ${\mathbf R}_{13} = 0$, notice that

$ {\mathbf R} = \left( \begin{array}{ccc}
                  {\mathbf R}_{11} & {\mathbf R}_{12} & {\mathbf R}_{13} \\
                  {\mathbf R}_{21} & {\mathbf R}_{22} & {\mathbf R}_{23}  \\
                  {\mathbf R}_{31} & {\mathbf R}_{32} & {\mathbf R}_{33}
                 \end{array} \right) $
corresponds to the decomposition of the CRN to three subnetworks $
N_1 = \lcub Y_1, Y_2, \ldots, Y_{t_1} \rcub, N_2 = \lcub
Y_{t_1+1}, Y_{t_1+2}, \ldots, Y_{t_1+t_2} \rcub, N_3 = \\ = \lcub
Y_{t_1+t_2+1}, Y_{t_1+t_2+2},\ldots,Y_m \rcub $. ${\mathbf
R}_{ii}$ corresponds to the relation matrix for $N_i, i=1,2,3$,
(Although $N_i$ together with ${\mathbf R}_{ii}$ may not
correspond to a weakly reversible chemical reaction network) while
${\mathbf R}_{ij}, i \ne j \in \lcub 1, 2, 3 \rcub$ gives the
information about reaction vectors from $N_i$ to $N_j$.

Now ${\mathbf R}_{13} = 0 $ means that each $Y_i$ belonging to
$N_1$ must pass $N_2$ to get to $Y_j$ which belongs to $N_3$. So
we can always find a ``reaction chain''

$$ Y_{i_1} \te Y_{i_2} \te \ldots \te Y_{i_s}, s \ge 3 $$

\noindent such that $ Y_{i_1} \in N_1$, $Y_{i_s} \in N_3$, and $
Y_{i_k} \in N_2$ for $ k \in \lcub 2, \ldots, s-1 \rcub$, which
means that $a_{i_1i_2}, a_{i_2i_3}, \ldots, a_{i_{s-1}i_s} > 0$
with

$$ a_{i_1i_2} \in {\mathbf R}_{12}, a_{i_2i_3}, a_{i_3i_4}, \ldots, a_{i_{s-2}i_{s-1}} \in {\mathbf R}_{22}, a_{i_{s-1}i_s} \in {\mathbf R}_{23}. $$

So now we have

\begin{equation}\label{key_est2}
\left .
\begin{array}{ll}
term (1) + (2) + (3) \le  a_{i_1i_2}y_{i_2} + e^{y_{i_2}}a_{i_2i_3}y_{i_3} + \ldots + e^{y_{i_{s-1}}} a_{i_{s-1}i_s} (-L) + \frac{m}{e} \lab {\mathbf R} \rab \\
 \le \tau({\mathbf R}) [ y_{i_2} + e^{y_{i_2}}y_{i_3} + e^{y_{i_3}}y_{i_4} + \ldots + e^{y_{i_{s-2}}}y_{i_{s-1}} + e^{y_{i_{s-1}}}(-L)] + \frac{m}{e} \lab {\mathbf R} \rab \\
\le  \tau({\mathbf R}) \cdot A_{s-2, L} + \frac{m}{e} \lab {\mathbf R} \rab ,
\end{array} \right .
\end{equation}

\end{itemize}

\noindent where in the last inequality of (\ref{key_est2}) we have used lemma 3.2.

\noindent Combining (\ref{key_est1}), (\ref{key_est2}) we see that

\begin{equation}\label{final_est}
 \max_{y \in B_L} F(y) \le \max \lcub
\frac{m}{e}\lab {\mathbf R} \rab - L \cdot \tau({\mathbf R}) ,
 \tau({\mathbf R}) \cdot
A_{s-2, L} + \frac{m}{e} \lab {\mathbf R} \rab \rcub,
\end{equation}

\noindent thus $C_L = \max_{y \in B_L} F(y) \te -\infty$ as $L \te
+\infty$. The proof of the ${\bf Claim}$ is now complete, and the
proof of Lemma 3.3 follows from that of the ${\bf Claim}$.

\end{proof}

\begin{remark}\rm This lemma was obtained by Adrian Nachman, around 1983. Using this lemma, Nachman was able to show the existence of positive equilibrium points for weakly reversible chemical reaction networks with one linkage class. After this Feinberg showed the existence of equilibrium points when the CRN is weakly reversible and has the property that the deficiency of the CRN is equal to the sum of the deficiencies of each linkage class subnetwork.
\end{remark}

\begin{corollary} If $ G(z) = e^z \cdot {\mathbf R}$ where $\mathbf R$ is the relation matrix for some CRN of only one linkage class, then for all $\tilde C > 0$, there exists $r(\tilde C,{\mathbf R}) > 0$ such that for all $w \in P_m$, $\lab w \rab_2 \le \tilde C$, we have

$$ (G(z+w),z) < 0, \hspace{0.2cm}  z \in P_m, \lab z \rab_2 > r(\tilde C,{\mathbf R}).$$

\end{corollary}

\begin{proof} First observe that

$$ (G(z+w),z) = e^{z+w} \cdot {\mathbf R} \cdot z^t = e^z \cdot \tilde{\mathbf R}(w) \cdot z^t $$

\noindent where $\tilde{\mathbf R}(w) = Diag(e^{w_1}, e^{w_2},
\ldots, e^{w_m}) \cdot {\mathbf R} $, with $ Diag(e^{w_1},
e^{w_2}, \ldots, e^{w_m})$ being the $m \times m$ diagonal matrix.
Thus $\tilde{\mathbf R}(w)$ is still a relation matrix of some CRN.

Think about $\tilde{\mathbf R}(w) = Diag(e^{w_1}, e^{w_2}, \ldots,
e^{w_m}) \cdot {\mathbf R} $ as a perturbation of ${\mathbf R}$,
in the class of relation matrix. Fix  $\tilde C > 0$, then for $w \in
\lcub w \in P_m, \lab w \rab_2 \le \tilde C \rcub$, the proof of Lemma 3.3 still
holds for $\tilde{\mathbf R}(w)$, specifically the estimate (\ref{final_est})
holds for $w \in P_m, \lab w \rab_2 \le \tilde C$, with ${\mathbf R}$ replaced by
$\tilde{\mathbf R}(w) $. Now observe that

$$ \min_{\lab w \rab_2 \le \tilde C} \tau(\tilde{\mathbf R}(w)) > 0, $$

\noindent then due to the compactness of $\lcub w \in P_m, \lab w \rab_2 \le \tilde C \rcub$ we can choose $r(\tilde C, \mathbf R) > 0$
large enough such that

$$ (G(z+w), z) = (e^z \cdot  \tilde{\mathbf R}(w), z) < 0 $$

\noindent holds for $w \in \lcub w \in P_m, \lab w \rab_2 \le \tilde C \rcub $
uniformly, where $z \in P_m, \lab z \rab_2 > r(\tilde C, \mathbf R)$.
\end{proof}

\begin{corollary} If $ G(z) = e^z \cdot {\mathbf R}$ where $\mathbf R$ is the relation matrix for some CRN of only one linkage class, then there exists constant $t_0({\mathbf R}) > 0 $ such that

$$ (G(z), z) < t_0({\mathbf R}), \hspace{0.2cm}  z \in P_m$$

\end{corollary}

\begin{proof} From Lemma $3.3$ we know that
there exists $ r_0 > 0 $, such that for all $ z \in P_m, \lab z
\rab > r_0$, we have $(G(z), z) < 0$. Let $ t_0 = \max(0, \max_{z
\in P_m, \lab z \rab \le r_0} \lab (G(z), z) \rab) $. Since
$(G(z), z)$ is continuous for $z \in B_{P_m}(r_0) = \lcub z \in
P_m : \lab z \rab \le r_0 \rcub$, we have $0 < t_0 < +\infty $.
\end{proof}

From Lemma $3.3$ we obtain easily

\begin{lemma} If the CRN has only one linkage class, then for any subspace $H$ of $P_m$, we have that

\begin{equation}
 G(H) \cap H^{\bot} \ne \emptyset
\end{equation}

\end{lemma}

\begin{proof} Since $R^m = H \oplus H^{\bot}$, we
can define the projection operator $\Pi_H: R^m \te H$ as

\begin{equation}
 \Pi_H(z) = z_1
\end{equation}

\noindent where $ z = z_1 + z_2$, $z_1 \in H, z_2 \in H^{\bot}$.

Then $G(H) \cap H^{\bot} \ne \emptyset$ if and only if $0 \in
\Pi_H(G(H))$. From (\ref{decay}) we know that for $r$ sufficiently large,

$$ (\Pi_H \circ G(z), z) = (G(z), z) \le -c(\lab z \rab)e^{\lab z \rab} < 0 $$
\noindent for $ z \in S_{H}(r) = \lcub z \in H : \lab z \rab = r
\rcub$.

By Brouwer's Fixed Point Theorem, there exists  $z_0 \in B_H(r) =
\lcub z \in H: \lab z \rab < r \rcub$ such that $\Pi_H \circ
G(z_0) = 0$. Thus we have

$$ 0 \in \Pi_H(G(H)) $$

\noindent which means $G(H) \cap H^{\bot} \ne \emptyset$.
\end{proof}

\section{The case of $l$ linkage classes}

 In this section we will discuss the case
of $l$ linkage classes. Without loss of generality we assume the
relation matrix $\mathbf{R}$ has block form (\ref{block}), and each $\mathbf
R_i$ is a $m_i \times m_i$ matrix for $i = 1, 2, \ldots, l$.

The content of this section is divided into two parts: the first
part is about the estimate of $(G(z),z)$. Similar to the case of
one linkage class, there exists an $l$ dimensional degenerate subspace
$D_l$ for $G$, i.e, $(G(z),z) = 0$ restricted to $D_l$. We will give an
estimate of $(G(z),z)$ restricted to the orthogonal subspace
$P_{m,l}$ of $D_l$ (Lemma 4.3), and correspondingly show that $G(H) \cap H^{\bot} \ne \emptyset$ when
$H \subset P_{m,l}$ (Lemma 4.4).

\bigskip

\begin{definition} Let $P_{m,l} = \lcub z \in R^m : \sum_{i =
1}^{m_1}z_i = 0, \sum_{i = m_1 + 1}^{m_1+m_2}z_i = 0, \ldots,
\sum_{i = m - m_l + 1}^{m}z_i = 0 \rcub $. Since $P_{m,l} \subset P_m$, $\lab z \rab$ is still a
norm for $z \in P_{m,l}$.

For $z = ( z^{(1)},  z^{(2)}, \dots,  z^{(l)}) \in P_{m,l}$ where $$  z^{(1)} = (z_1, z_2, \dots, z_{m_1}),  z^{(2)} = (z_{m_1+1}, z_{m_1+2}, \dots, z_{m_1+m_2}), \dots, z^{(l)} = (z_{m-m_l+1}, z_{m-m_l+2}, \dots, z_m),$$ we define $\Pi_i: P_{m,l} \te P_{m_i} $ as

$$ \Pi_i(z) = z^{(i)}, \hspace{.2in} i = 1,2, \dots, l. $$

\end{definition}

\begin{remark} By Lemma 3.3,  there exists $ 0 < R_0 < +\infty$ such that

\begin{equation}\label{uniform_est}
 (e^{z^{(i)}}{\mathbf R_i}, z^{(i)}) < 0
\end{equation}

\noindent for all $z^{(i)} \in P_{m_i}, \lab z^{(i)} \rab_2 > R_0,
i =1,2,\ldots,l$. \end{remark}

We have the following

\begin{lemma} If the CRN has exactly $l$ linkage classes, then there exists a function $c^*(r)$ on $R^1$ with $c^*(r) \te +\infty$ as $r \te +\infty$ such that

\begin{equation}\label{est_l}
 (G(z), z) \le - c^*(\lab z \rab)e^{\lab z
\rab}
\end{equation}

\noindent for all $z \in P_{m,l}$.

\end{lemma}

\begin{proof} Using the estimate in Lemma 3.3 for one linkage class we have

\begin{equation}
\left \{
\begin{array}{ll}
(G(z),z) = (e^{\underline z^{(1)}} \hspace{0.1cm} e^{\underline
z^{(2)}} \hspace{0.1cm} \ldots \hspace{0.1cm}
e^{\underline z^{(l)}}) {\mathbf R} \left( \begin{array}{c}
{\underline z^{(1)}}^t \\ {\underline z^{(2)}}^t \\ \vdots \\
{\underline z^{(l)}}^t \end{array} \right)
 =  e^{\underline z^{(1)}}{\mathbf R_1}{\underline z^{(1)}}^t + e^{\underline z^{(2)}}{\mathbf R_2}{\underline z^{(2)}}^t + \ldots  \\ + \hspace{0.1cm} e^{\underline z^{(l)}}{\mathbf R_l}{\underline z^{(l)}}^t
 \le - \min \lcub c_1(\lab z \rab), c_2(\lab z \rab), \ldots, c_l(\lab z \rab) \rcub e^{\lab z \rab} + \sum_{i=1}^l t_i \le -c^*(\lab z \rab)e^{\lab z \rab}
\end{array} \right .
\end{equation}

\noindent where $c^*(r) = \min \lcub c_1(r), c_2(r), \ldots,
c_l(r) \rcub - \sum_{i=1}^l t_i$ , with $c_i(r), t_i$
corresponding to the relation matrix ${\mathbf R_i}, i = 1, 2,
\ldots, l$, respectively. \end{proof}

From Lemma 4.3 we obtain easily

\begin{lemma} If the CRN has exactly $l$ linkage classes, then for any subspace $H$ of $P_{m,l}$, we have that

\begin{equation}
 G(H) \cap H^{\bot} \ne \emptyset
\end{equation}

\end{lemma}

\begin{proof} Since $R^m = H \oplus H^{\bot}$, we
can define the projection operator $\Pi_H: R^m \te H$ as

\begin{equation}
 \Pi_H(z) = z_1
\end{equation}

\noindent where $ z = z_1 + z_2$, $z_1 \in H, z_2 \in H^{\bot}$.

It follows that $G(H) \cap H^{\bot} \ne \emptyset$ iff $0 \in
\Pi_H(G(H))$. From (\ref{est_l}) we know that for $r$ sufficiently large,

$$ (\Pi_H \circ G(z), z) = (G(z), z) \le -c^*(\lab z \rab)e^{\lab z \rab} < 0 $$

\noindent for $ z \in S_{H}(r) = \lcub z \in H : \lab z \rab = r
\rcub$.

By the Brouwer Fixed Point Theorem, there exists  $z_0 \in B_H(r)
= \lcub z \in H: \lab z \rab < r \rcub$ such that $\Pi_H \circ
G(z_0) = 0$. Thus we have

$$ 0 \in \Pi_H(G(H)) $$

\noindent which means $G(H) \cap H^{\bot} \ne \emptyset$.
\end{proof}

\bigskip

Now we turn to the second part of this section. We want to discuss
an \emph{arbitrary} subspace $K$ of $R^m$ and show that (\ref{intersec}) is
true. The natural idea is to transfer the intersection problem for
$K$ to that for $P_{m,l}$. Therefore the possible configuration of
$K$ in $R^m = D_l \oplus P_{m,l}$ is first discussed, then via
Lemma $4.5$ the intersection problem (\ref{intersec}) is transformed to (\ref{key_lem}),
which is an intersection problem with respect to $P_{m,l}$. This
is achieved at the expense of adding an \emph{arbitrary} linear
map $F$ from some subspace of $P_{m,l}$ to $D_l$, which leads to
the change of relation function $G$ to that of $G^*$.

But the essence of this section is actually Lemma $4.5$, so a few
words about the idea of proof are helpful. The idea is to construct a domain $\Omega_H(r) \subset
H$ which resembles the ball $B_H(r)$ with a slight modification.
We want to choose $\Omega_H(r)$ so that it is compact and convex
and thus homeomorphic to the ball. Moreover, we want to show that
on the boundary of $\Omega_H(r)$ we have

$$ (\Pi_H \circ G^*(z), \hat n(z)) = (G^*(z), \hat n(z)) < 0 $$

\noindent where $z \in \pa \Omega_H(r)$, $\hat n(z)$ means the outward normal of $\pa \Omega_H(r)$ at $z$. 

Then the Brouwer Fixed Point Theorem can be applied to obtain
Lemma $4.5$.

\bigskip

\begin{lemma} For any subspace $ H \subset P_{m,l}$, and for any linear map $F: H \te D_l$, we have

\begin{equation}\label{key_lem}
  G \circ (I+F)(H) \cap H^{\bot} \ne \emptyset
\end{equation}

\noindent or, equivalently, there exists $x \in H$, such that $ G(x+F(x)) \in H^{\bot} $.

\end{lemma}


\begin{proof} For the sake of simplicity, we give the proof for $l=2$. The proof for $l \ge 3$ case is similar.

When $l = 2$, we define

\begin{equation}
H_1 =  H \cap \Pi_1^{-1}(0), \hspace{.1in} H_2 =  H \cap \Pi_2^{-1}(0), \hspace{.1in} H_3 =  H \cap (H_1+H_2)^{\bot}
\end{equation}

\noindent where $\Pi_i: P_{m, l} \te P_{m_i}, i = 1, 2$ is as given in Definition 4.1. Then $H = H_1 \oplus H_2 \oplus H_3$ and $H_i \bot H_j$ for $i \ne
j \in \lcub 1,2,3 \rcub$.


Thus given $ z \in H$, we have

\begin{equation}
 z = z_1 + z_2 + z_3 = ( 0,
 x_1) + ( x_2,  0) + (
y_1, y_2)
\end{equation}

\noindent where $ z_1 = (0,x_1) \in H_1$, $z_2 = (x_2,0) \in H_2$,
$z_3 = (y_1, y_2) \in H_3 $. Then there exists $ C_1 > C_2 > 0 $ such that  $$ C_2 \lab y_2 \rab_2 \le \lab y_1 \rab_2 \le C_1 \lab y_2 \rab_2$$ for all $ z_3 = (y_1, y_2) \in H_3$, since by construction $\Pi_i$ restricted to $H_3$ are isomorphisms for $i = 1, 2$.

\bigskip

\noindent Now define $\Omega_H(r) = \lcub z \in H : \lab z \rab_2
\le r, \lab z_1 \rab_2 \le r - \frac{C^{(2)}}{r}, \lab z_2 \rab_2
\le r - \frac{C^{(2)}}{r}, \lab z_3 \rab_2 \le r -
\frac{C^{(2)}}{r} \rcub $, where $ C^{(2)} = 3\max \lcub 1, R_0^2
\cdot \frac{1+ C_2^2}{C_2^2}, R_0^2 \cdot (1+C_1^2) \rcub $ is a
constant.

\bigskip

First, we have the following

\begin{claim} If $r > C^{(2)}$, then we have that
$\Omega_H(r)$ is homeomorphic to $ \bar B_H(1) = \lcub z \in H : \lab z
\rab_2 \le 1 \rcub$; thus $\partial \Omega_H(r)$ is homeomorphic
to $S_H(1) = \partial B_H(1) = \lcub z \in H : \lab z \rab = 1
\rcub$. \end{claim}

\noindent Proof of Claim 4.6: $\Omega_H(r)$ is homeomorphic to
$\bar B_H(1)$ due to the fact that for $r > C^{(2)}$, $\Omega_H(r)
\subset H$ is a compact and convex set, with nonempty interior, thus is
homeomorphic to $\bar B_H(1)$. The second claim follows from the fact
that the restriction of a homeomorphism is also a homeomorphism. Q.E.D.

\begin{notation} We define the \emph{structure function} $G^*: H \te P_{m,l}$ as

\begin{equation}
 G^*(x) = G(x+F(x)).
\end{equation}

\noindent For the linear map $F: H \te D_l$, we will also write $F$ as
$$ F(z) = (F_1(z)\underline 1_{m_1}, F_2(z)\underline 1_{m_2},
\ldots, F_l(z)\underline 1_{m_l}),$$  where $F_i(z): H \te R$ is a
linear function for $i = 1,2,\ldots, l$.
\end{notation}

Next we show that

\begin{claim} For $r > 0$ sufficiently large, we have $
(\Pi_H \circ G^*(z), \hat n(z)) < 0$, for all $z \in \partial
\Omega_H(r)$, where $\hat n(z)$ means the outer normal of
$\Omega_H(r)$ at $z \in \partial \Omega_H(r)$.\end{claim}

\noindent Proof of Claim 4.7: When $r > 0$ is sufficiently large, we
have $\partial \Omega_H(r) = S_1 \cup S_2 \cup S_3 \cup S_4$,
where $ \left \{
      \begin{array}{ll}
      S_1 = \lcub z \in H: \lab z \rab_2 = r, \lab z_1 \rab_2 \le r - \frac{C^{(2)}}{r}, \lab z_2 \rab_2 \le r - \frac{C^{(2)}}{r}, \lab z_3 \rab_2 \le r - \frac{C^{(2)}}{r} \rcub, \\
      S_2 = \lcub z \in H: \lab z_1 \rab_2 = r - \frac{C^{(2)}}{r}, \lab z \rab_2 \le r, \lab z_2 \rab_2 \le r - \frac{C^{(2)}}{r}, \lab z_3 \rab_2 \le r - \frac{C^{(2)}}{r} \rcub, \\
      S_3 = \lcub z \in H: \lab z_2 \rab_2 = r - \frac{C^{(2)}}{r}, \lab z_1 \rab_2 \le r - \frac{C^{(2)}}{r}, \lab z \rab_2 \le r, \lab z_3 \rab_2 \le r - \frac{C^{(2)}}{r} \rcub, \\
      S_4 = \lcub z \in H: \lab z_3 \rab_2 = r - \frac{C^{(2)}}{r}, \lab z_1 \rab_2 \le r - \frac{C^{(2)}}{r}, \lab z_2 \rab_2 \le r - \frac{C^{(2)}}{r}, \lab z \rab_2 \le r \rcub
      \end{array} \right. $

\noindent with outer normal $\hat n(z) = \frac{z}{\lab z \rab_2},
\frac{z_1}{\lab z_1 \rab_2}, \frac{z_2}{\lab z_2 \rab_2},
\frac{z_3}{\lab z_3 \rab_2}$, respectively.

Thus for $ z \in S_1$, we have

\begin{equation}\label{norm_est11}
\begin{array}{ll}
(\Pi_H \circ G^*(z), \hat n(z)) =  (\Pi_H \circ G^*(z), \frac{z}{\lab z \rab_2}) = (G^*(z), \frac{z}{\lab z \rab_2}) \\
= \frac{e^{z+F(z)}{\mathbf R}z}{\lab z \rab_2} =
\frac{e^{x_2+y_1+F_1(z)\underline 1_{m_1}}{\mathbf R_1}(x_2+y_1)^t
+ e^{x_1+y_2+F_2(z)\underline 1_{m_2}}{\mathbf R_2}(x_1+y_2)^t
}{\lab z \rab_2},
\end{array} 
\end{equation}

Now from $z \in S_1$ we have $\lab z \rab_2 = r, \lab z_1 \rab_2
\le r - \frac{C^{(2)}}{r}, \lab z_2 \rab_2 \le r -
\frac{C^{(2)}}{r}$, thus

\begin{equation}\label{norm_est}
\left \{
\begin{array}{ll}
{\lab z_2 \rab_2}^2 + {\lab z_3 \rab_2}^2 = & {\lab z \rab_2}^2 - {\lab z_1 \rab_2}^2 \ge r^2 - (r - \frac{C^{(2)}}{r})^2 = 2C^{(2)} - \frac{{C^{(2)}}^2}{r^2} > C^{(2)} \\
{\lab z_1 \rab_2}^2 + {\lab z_3 \rab_2}^2 = & {\lab z \rab_2}^2 -
{\lab z_2 \rab_2}^2 \ge r^2 - (r - \frac{C^{(2)}}{r})^2 = 2C^{(2)}
- \frac{{C^{(2)}}^2}{r^2} > C^{(2)}
\end{array} \right.
\end{equation}

\noindent for $r > C^{(2)}$.

\begin{itemize}

\item If $\lab z_3 \rab_2^2 \le \frac{C^{(2)}}{2}$, then from (\ref{norm_est})
we have

\begin{equation}
\left \{
\begin{array}{ll}
{\lab x_1 + y_2 \rab_2}^2 \ge {\lab x_1 \rab_2}^2 = \lab z_1 \rab^2 > \frac{C^{(2)}}{2} > {R_0}^2 \\
{\lab x_2 + y_1 \rab_2}^2 \ge {\lab x_2 \rab_2}^2 = \lab z_2
\rab^2 > \frac{C^{(2)}}{2} > {R_0}^2
\end{array} \right.
\end{equation}

\noindent since $H_1 \bot H_3, H_2 \bot H_3$. Thus by Remark 4.2 and (\ref{norm_est11}) we have

\begin{equation}\label{norm_est1_res}
\left \{
\begin{array}{ll}
(\Pi_H \circ G^*(z),z) = \frac{e^{x_2+y_1+F_1(z) \underline 1_{m_1}}{\mathbf R_1}(x_2+y_1)^t + e^{x_1+y_2+F_2(z) \underline 1_{m_2}}{\mathbf R_2}(x_1+y_2)^t }{\lab z \rab_2} \\
= \frac{e^{F_1(z)} \cdot e^{x_2+y_1}{\mathbf R_1}(x_2+y_1)^t +
e^{F_2(z)} \cdot e^{x_1+y_2}{\mathbf R_2}(x_1+y_2)^t }{\lab z
\rab_2} < 0
\end{array} \right.
\end{equation}

\item If ${\lab z_3 \rab_2}^2 > \frac{C^{(2)}}{2}$, then we have

\begin{equation}
\left \{
\begin{array}{ll}
{\lab y_1 \rab_2}^2 \ge \frac{C_2^2}{1+ C_2^2} {\lab z_3 \rab_2}^2 > \frac{C_2^2}{1+ C_2^2} \cdot \frac{C^{(2)}}{2} > R_0^2 \\
{\lab y_2 \rab_2}^2 \ge \frac{1}{1+ C_1^2} {\lab z_3 \rab_2}^2 >
\frac{1}{1+ C_1^2} \cdot \frac{C^{(2)}}{2} > R_0^2
\end{array} \right.
\end{equation}

\noindent Thus, similarly, we also have that (\ref{norm_est1_res}) holds.

\end{itemize}

Summarizing the estimate above, we have shown that if $r > C^{(2)}$, then

\begin{equation}\label{norm_est1}
 (\Pi_H \circ G^*(z), \hat n(z) ) < 0
\end{equation}

\noindent for all $z \in S_1$.

For $ z \in S_2$, we have

\begin{equation}
\left \{
\begin{array}{ll}
(\Pi_H \circ G^*(z), \hat n(z)) =  (\Pi_H \circ G^*(z), \frac{z_1}{\lab z_1 \rab_2}) = (G^*(z), \frac{z_1}{\lab z_1 \rab_2}) \\
= \frac{e^{z+F(z)}{\mathbf R}z_1^t}{\lab z_1 \rab_2} = \frac{
e^{x_1+y_2+F_2(z)\underline 1_{m_2}}{\mathbf R_2}(x_1)^t }{\lab
z_1 \rab_2}
\end{array} \right.
\end{equation}

Now, from $z \in S_2$ we have $\lab z_1 \rab_2 = r -
\frac{C^{(2)}}{r}, \lab z \rab_2 \le r$, thus we have

\begin{equation}
\left \{
\begin{array}{ll}
{\lab y_2 \rab_2}^2 \le \lab z_3 \rab_2^2 \le \lab z \rab_2^2 - \lab z_1 \rab_2^2 \le r^2 - (r - \frac{C^{(2)}}{r})^2 \\
 = 2C^{(2)} - \frac{{C^{(2)}}^2}{r^2} < 2C^{(2)}
\end{array} \right.
\end{equation}

\noindent which means $ \lab y_2 \rab_2 < \sqrt{2C^{(2)}}$.

With Corollary 3.5 of Section 3 in mind, we choose $r_1 > 0$ satisfying
$r_1 - \frac{C^{(2)}}{r_1} > r(\sqrt{2C^{(2)}},{\mathbf R_2}))$,
such that when $r > r_1$, we have $\lab x_1 \rab_2 = \lab z_1
\rab_2 = r - \frac{C^{(2)}}{r} > r_1 - \frac{C^{(2)}}{r_1}$, thus
combining the two inequalities above we have

\begin{equation}
\frac{e^{x_1+y_2+F_2(z) \underline 1_{m_2}}{\mathbf R_2}(x_1)^t }{\lab z_1 \rab_2}
= \frac{e^{F_2(z)} \cdot e^{x_1+y_2}{\mathbf R_2}(x_1)^t }{\lab
z_1 \rab_2} < 0
\end{equation}

Thus we have shown that if $r > r_1$, then

\begin{equation}\label{norm_est2}
 (\Pi_H \circ G^*(z), \hat n(z) ) < 0
\end{equation}

\noindent for all $z \in S_2$.

Similarly for $z \in S_3$, we can choose $r_2 > 0$ satisfying $r_2
- \frac{C^{(2)}}{r_2} > r(\sqrt{2C^{(2)}},{\mathbf R_1}))$, such
that when $ r > r_2$, we have

\begin{equation}\label{norm_est3}
(\Pi_H \circ G^*(z), \hat n(z) ) = \frac{e^{x_2+y_1+F_1(z) \underline 1_{m_1}}{\mathbf R_1}(x_2)^t }{\lab z_2 \rab_2}
= \frac{e^{F_1(z)} \cdot e^{x_2+y_1}{\mathbf R_1}(x_2)^t }{\lab
z_2 \rab_2} < 0
\end{equation}

\noindent for all $z \in S_3$.

Now for $ z \in S_4$, we have

\begin{equation}
\begin{array}{ll}
(\Pi_H \circ G^*(z), \hat n(z)) = (\Pi_H \circ G^*(z), \frac{z_3}{\lab z_3 \rab_2}) = (G^*(z), \frac{z}{\lab z \rab_2}) \\
= \frac{e^{z+F(z)}{\mathbf R}z_3^t}{\lab z_3 \rab_2} =
\frac{e^{x_2+y_1+F_1(z)\underline 1_{m_1}}{\mathbf R_1}(y_1)^t +
e^{x_1+y_2+F_2(z)\underline 1_{m_2}}{\mathbf R_2}(y_2)^t }{\lab
z_3 \rab_2}
\end{array} 
\end{equation}

$z \in S_4$ implies $\lab z_3 \rab_2 = r -
\frac{C^{(2)}}{r}, \lab z \rab_2 \le r$, thus

\begin{equation}
\left \{
\begin{array}{ll}
\lab x_1 \rab_2^2 = \lab z_1 \rab_2^2 \le \lab z \rab_2^2 - {\lab z_3 \rab_2}^2 \le r^2 - (r - \frac{C^{(2)}}{r})^2 = 2C^{(2)} - \frac{{C^{(2)}}^2}{r^2} < 2C^{(2)}, \\
\lab x_2 \rab_2^2 = \lab z_2 \rab_2^2 \le \lab z \rab_2^2 - {\lab
z_3 \rab_2}^2 \le r^2 - (r - \frac{C^{(2)}}{r})^2 = 2C^{(2)} -
\frac{{C^{(2)}}^2}{r^2} < 2C^{(2)}
\end{array} \right.
\end{equation}

\noindent which means that $ \lab x_1 \rab_2 < \sqrt{2C^{(2)}}$, $
\lab x_2 \rab_2 < \sqrt{2C^{(2)}}$.

By Corollary $3.5$ of Section 3, we can choose $r_3 > 0$
satisfying $r_3 - \frac{C^{(2)}}{r_3} > \max (
\sqrt{\frac{1+C_2^2}{C_2^2}}, \sqrt{1+C_1^2} ) \cdot \max (
r(\sqrt{2C^{(2)}},{\mathbf R_1}), r(\sqrt{2C^{(2)}}, {\mathbf
R_2}) )$, such that when $r > r_3$, we have $\lab z_3 \rab_2 = r -
\frac{C^{(2)}}{r} > r_3 - \frac{C^{(2)}}{r_3}$, which leads to

\begin{equation}
\left \{
\begin{array}{ll}
\lab y_1 \rab_2 \ge \sqrt{\frac{C_2^2}{1+C_2^2}} \lab z_3 \rab_2 > r(\sqrt{2C^{(2)}},{\mathbf R_1}) \\
\lab y_2 \rab_2 \ge \sqrt{\frac{1}{1+C_1^2}} \lab z_3 \rab_2 >
r(\sqrt{2C^{(2)}},{\mathbf R_2})
\end{array} \right.
\end{equation}

Thus, we have that when $r > r_3$,

\begin{equation}\label{norm_est4}\left \{ \begin{array}{ll}
(\Pi_H \circ G^*(z), \hat n(z) ) = \frac{e^{x_2+y_1+F_1(z)\underline 1_{m_1}}{\mathbf R_1}(y_1)^t + e^{x_1+y_2+F_2(z)\underline 1_{m_2}}{\mathbf R_2}(y_2)^t }{\lab z_3 \rab_2} \\
= \frac{e^{F_1(z)} \cdot e^{x_2+y_1}{\mathbf R_1}(y_1)^t +
e^{F_2(z)} \cdot e^{x_1+y_2}{\mathbf R_2}(y_2)^t }{\lab z_3
\rab_2} < 0 \end{array}\right.
\end{equation}

\noindent for all $z \in S_4$.

So combining (\ref{norm_est1}), (\ref{norm_est2}), (\ref{norm_est3}) and (\ref{norm_est4})  we have
finished the proof of Claim 4.7. Q.E.D.

\bigskip

Combining Claim 4.6 and Claim 4.7, also noticing the discussion before
Lemma $4.5$,  we finish the proof of Lemma $4.5$.

\end{proof}

\begin{remark} It is easy to see that $\Pi_H \circ G^* $ is an analytic function defined on $H$, so there can only be finitely many $z \in \Omega_H(r) \subset H$ such that $\Pi_H \circ G^*(z) = 0$. Also from the proof of Lemma $4.5$ we can see that outside $\Omega_H(r)$ there is no $z \in H$ belonging to $G^*(H) \cap H^{\bot}$.
\end{remark}

\bigskip

\begin{remark} It is an interesting observation that in the proof of Lemma $4.5$, we \emph{never} use the fact that $F: H \te D_l$ is linear. In fact, the reader may readily check that if we change the word ``linear'' to ``nonlinear'' in Lemma $4.5$, the proof of Lemma $4.5$ still goes through. This observation will play a fundamental role in the proof of Theorem 1.1. We write is as the following
\end{remark}

\begin{corollary}For all $ H \subset P_{m,l}$ and for all \emph{continuous} map $F: H \te D_l$, we have

\begin{equation}
 G^*(H)\cap H^{\bot} = G \circ (I+F)(H) \cap
H^{\bot} \ne \emptyset
\end{equation}

or, equivalently, $ G(x+F(x)) \in H^{\bot} $ for some $x \in H$.
\end{corollary}

\begin{lemma} For all $ H' \subset R^n$ such that $H' \cap D_l = \lcub \theta \rcub$, we have

$$ G(H') \cap {H'}^{\bot} \ne \emptyset. $$

\end{lemma}

\begin{proof} Let $H = \Pi_{P_{m,l}}(H')$. From
the assumption $H' \cap D_l = H' \cap P_{m,l}^{\bot} = \lcub
\theta \rcub$ we know that $\Pi_{P_{m,l}}: H' \te H$ is one-to-one
and onto, thus is an isomorphism. It then follows that there exists a linear map $F: H
\te D_l$ such that for all $z \in H'$, there exists  unique $x \in H$
such that

\begin{equation}
 z = x + F(x)
\end{equation}

\bigskip

\begin{claim} $G(H') \cap {H'}^{\bot} = G(H') \cap H^{\bot}$ \end{claim}

\noindent Proof of {\bf Claim 4.12}: For all $y \in G(H') \cap
{H'}^{\bot}$, $ z \in H'$, we have $y \in G(H') \subset P_{m,l}$,
thus

\begin{equation}
 (y,z) = (y, x+F(x)) = (y, x) = 0,
\end{equation}

\noindent which shows that $y \in H^{\bot}$, so we obtain

\begin{equation}\label{inclu_1}
 G(H')\cap {H'}^{\bot} \subset G(H') \cap
H^{\bot}
\end{equation}

Now, for all $y' \in G(H')\cap H^{\bot}$, $ x \in H$, we have

\begin{equation}
 (y',x) = (y', x+F(x)) = (y',z) = 0
\end{equation}

Thus

\begin{equation}\label{inclu_2}
 G(H')\cap H^{\bot} \subset G(H') \cap
{H'}^{\bot}
\end{equation}

Combining (\ref{inclu_1}) and (\ref{inclu_2}), we obtain

\begin{equation}
 G(H')\cap {H'}^{\bot} = G(H') \cap H^{\bot}.
\end{equation}

Q.E.D. \bigskip

Now utilizing Claim 4.12 we obtain

\begin{equation}
G(H') \cap {H'}^{\bot} = G(H') \cap H^{\bot} = G \circ (I + F)(H) \cap H^{\bot} \ne \emptyset
\end{equation}

\noindent where the last statement follows from Lemma 4.5.

\end{proof}

\begin{lemma} If the CR has exactly $l$ linkage classes, then for any subspace $H$ of $R^m$, we have that

\begin{equation}
G(H) \cap H^{\bot} \ne \emptyset
\end{equation}

\end{lemma}

\begin{proof} If the CR has $l$ linkage classes, by
relabeling $Y_1, Y_2, \ldots, Y_m$ we may assume that $\mathbf R$
has the same diagonal block form as that of (\ref{block}).

Now let $H \subset R^n$ be fixed. By linear algebra, we know
that $ H = H^{(1)} \oplus H^{(2)}$, where $H^{(1)} = H \cap D_l
\subset D_l, H^{(2)} \cap D_l = \lcub \theta \rcub$. Thus we have

\begin{equation}
\left \{
\begin{array}{ll}
G(H) \cap H^{\bot} = G(H^{(1)} \oplus H^{(2)}) \cap (H^{(1)} \oplus H^{(2)})^{\bot} = G(H^{(1)} \oplus H^{(2)}) \cap {H^{(1)}}^{\bot} \cap {H^{(2)}}^{\bot} \\
=  G(H^{(1)} \oplus H^{(2)}) \cap {H^{(2)}}^{\bot}
\end{array} \right.
\end{equation}

\noindent since combining $H^{(1)} \subset D_l = P_{m,l}^{\bot}$ and
$G(R^m) \subset P_{m,l}$ we have $G(H^{(1)}\oplus H^{(2)}) \subset
P_{m,l} \subset {H^{(1)}}^{\bot}$. Finally noting that

$$  G( H^{(2)}) \cap {H^{(2)}}^{\bot} \ne \emptyset $$ due to Lemma 4.9, we have

$$ G(H^{(1)} \oplus H^{(2)}) \cap {H^{(2)}}^{\bot} \supseteq G( H^{(2)}) \cap {H^{(2)}}^{\bot} \ne \emptyset. $$

\end{proof}

\section{Proof of Main Result and the index formula}

Now we only need to observe that the existence problem for each
compatibility class can be reduced to one of intersection problem type
as in the previous section.

\bigskip

\noindent Proof of {\bf Theorem 1.1}: First a few words about the notation
and definitions.

Let $C, \mathbf{R}, K, N$ be defined as in Section 2. Suppose that the CR has
$l$ linkage classes, by relabeling $Y_1, Y_2, \ldots, Y_m$ we may
assume that the relation matrix $\mathbf R$ has the same diagonal
block form as that of (\ref{block}). Let $S$ be the stoichiometric subspace
in $R^n$, and define $ W(x_0) = \ln((x_0+S) \cap R_+^n) \cdot C^t
\subset K$, where $x_0 \in R_+^n$. \bigskip

As in the proof of Lemma $4.13$, we have
the decomposition of $K$ as $K =
 K_1 \oplus K_2'$, where $K_1 = K \cap D_l \subset D_l, K_2' \cap D_l = \lcub \theta \rcub$. Now let $ K_2 = \Pi_{P_{m,l}}(K_2')$. Then similar to the proof of Lemma $4.11$ we have

\begin{equation}\label{nontriv_1}
 G(W(x_0)) \cap K^{\bot}  = G(W(x_0)) \cap {K_2'}^{\bot} =
G(W(x_0)) \cap {K_2}^{\bot}
\end{equation}

Next we will need the following lemma, for the proof, the reader is referred to Proposition {\bf B.1} in (\cite{F3}).

\begin{lemma} Given $ x_0 \in R_+^n$, we have $\Psi: (x_0+S) \cap R_+^n \te K_2 $ given by $\Psi(x) = \Pi_{K_2}(\ln{x} \cdot C^t)$ is one-to-one and onto, actually, a diffeomorphism.

\end{lemma}

Due to Lemma $5.1$ and the fact that $W(x_0) \subset K$,
$\Pi_{P_{m,l}}K =K_2$, there exists a continuous ( probably \emph{nonlinear})
map $F: K_2 \te D_l $ such that for all $x \in W(x_0)$, there
exists unique $ z \in K_2 $ such that $ x = z + F(z). $ Defining a function $\hat G: K_2 \te
P_{m,l}$ as

\begin{equation}
 \hat G(z) = G(z+F(z))
\end{equation}

\noindent for all $z \in K_2$. Note that $\hat G(K_2) = G(W(x_0))$.

Due to Remark $4.9$ (or Corollary $4.10$) of Section 4 we know that

\begin{equation}\label{nontriv_2}
\hat G(K_2) \cap {K_2}^{\bot} = G \circ (I + F)(K_2) \cap {K_2}^{\bot} \ne \emptyset,
\end{equation}

\noindent i.e., we have $  G(W(x_0)) \cap {K_2}^{\bot} \ne \emptyset $.

Combining (\ref{nontriv_1}) with (\ref{nontriv_2}) we have that $ G(W(x_0)) \cap K^{\bot} \ne
\emptyset $. Now we only need to observe that the map $F$ we
defined above is an analytic function, thus by Remark $4.8$ of Section
4, we must have $car(G(W(x_0)) \cap K^{\bot})$ is finite, where $car(Q)$ means the
cardinality of the set $Q$. The proof of Theorem 1.1 is complete.

\begin{corollary}Let $s$ be the dimension of the stiochiometric subspace. If we restrict the vector field $f$ in (\ref{dyn}) to some fixed positive stiochiometric compatibility class, denote the steady states in the stoichiometric compatibility class as $x_1, x_2, \ldots, x_t$ (the number $t$ of steady states depends on the stoichiometric compatibility class we choose), then we have

$$ \sum_{i=1}^t ind(x_i) = (-1)^s $$

\noindent where $ind(x_i)$ is the index of vector field $f$ at
steady state $x_i, i = 1, 2, \ldots, t$.
\end{corollary}

\begin{proof} Let $C, {\mathbf R}, K, K_2,
S$ be defined as above.

First, we give the following

\begin{definition} \rm For the configuration matrix $ C = \left( \begin{array}{c}
Y_1 \\ Y_2 \\ \vdots \\ Y_m \end{array} \right)_{m \times n}$, we
define the reduced configuration matrix $\tilde C$ as

$$ \tilde C = C -  \left( \begin{array}{cccc}
                  {\mathbf D}_1    \\
                   & {\mathbf D}_2  \\
                   & &  \ddots &     \\
                   & & & {\mathbf D}_l
                 \end{array} \right) \cdot C $$

\noindent where $D_i = \frac{1}{m_i} \cdot \left( \begin{array}{cccc} 1 & 1 & \ldots & 1 \\
   1 & 1 & \ldots & 1 \\
   \vdots & \vdots & \vdots & \vdots \\
   1 & 1 & \ldots & 1
   \end{array} \right)_{m_i \times m_i}$, with $m_i$ given at the beginning of section $4$ for $i = 1, 2, \ldots, l$.

\end{definition}

It is easy to check that $\tilde C$ has the following property

\begin{equation}\label{prop_tilde_c}
\left \{
\begin{array}{ll}
{\mathbf R} \cdot C = {\mathbf R} \cdot {\tilde C} \\
\Pi_{K_2}(x \cdot C^t) = \Pi_{K_2}(x \cdot {\tilde C}^t)
\end{array} \right.
\end{equation}

Now fix $x_0 \in R_+^n$, from the proof of Theorem 1.1 above we know
the vector field $f$ on $x_0 + S$ can be written as

\begin{equation}
 \frac{dx}{dt} = f(x) = e^{\ln{x} \cdot C^t}
\cdot {\mathbf R} \cdot C = e^{z+F(z)} \cdot {\mathbf R} \cdot
{\tilde C}
\end{equation}

\noindent where $z = \Psi(x) = \Pi_{K_2}(\ln{x} \cdot C^t) \in
K_2$. By Lemma $5.1$ we know that $\Psi: x_0+S \te K_2$ is a
diffeomorphism. Thus considering the induced vector field $\Psi_*
\circ f$ on $K_2$ we have

\begin{equation}
\left \{
\begin{array}{lll}
\frac{dz}{dt} & = \Psi_* \circ f \circ \Psi^{-1}(z) = \Pi_{K_2}(\frac{dx}{dt} \cdot Diag(\frac{1}{x}) \cdot C^t) \\
 & =  \Pi_{K_2}(e^{z+F(z)} \cdot {\mathbf R} \cdot {\tilde C} \cdot Diag(\frac{1}{x}) \cdot {\tilde C}^t) \\
 & =  \Pi_{K_2}(e^{z+F(z)} \cdot {\mathbf R}) \cdot {\tilde C}  \cdot Diag(\frac{1}{x}) \cdot {\tilde C}^t
\end{array} \right.
\end{equation}

\noindent where $Diag(\frac{1}{x}) = \left( \begin{array}{cccc}
                  \frac{1}{x_1}    \\
                   & \frac{1}{x_2}  \\
                   & &  \ddots &     \\
                   & & & \frac{1}{x_n}
                 \end{array} \right)$ and we have used property (\ref{prop_tilde_c}) of the reduced configuration matrix $\tilde C$.

Now notice that the matrix ${\tilde C} \cdot \left(
\begin{array}{cccc}
                  \frac{1}{x_1}    \\
                   & \frac{1}{x_2}  \\
                   & &  \ddots &     \\
                   & & & \frac{1}{x_n}
                 \end{array} \right) \cdot {\tilde C}^t$ is positive definite when the action is restricted to $K_2$, thus by a simple homotopy between ${\tilde C} \cdot \left( \begin{array}{cccc}
                  \frac{1}{x_1}    \\
                   & \frac{1}{x_2}  \\
                   & &  \ddots &     \\
                   & & & \frac{1}{x_n}
                 \end{array} \right) \cdot {\tilde C}^t$ and the identity matrix we see that the vector field $\Psi_* \circ f \circ \Psi^{-1}$ is homotopic to $\Pi_{K_2}(e^{z+F(z)} \cdot {\mathbf R})$, while the critical points and corresponding indices remain invariant under this homotopy.

We observe that in the proof of Lemma $4.5$ we construct a
bounded, convex domain $\Omega$ so that the vector field
$\Pi_{K_2}(e^{z+F(z)} \cdot {\mathbf R})$ is pointing inwards at
each point of $\partial \Omega$, thus by a standard result of
degree theory we have that

\begin{equation}
 \sum_{i=1}^t ind_f(x_i) = \sum_{i=1}^t
ind_{\Psi_* \circ f \circ \Psi^{-1}}(x_i(z)) = \sum_{i=1}^t
ind_{\Pi_{K_2}(e^{z+F(z)} \cdot {\mathbf R})}(\tilde x_i(z)) =
(-1)^s
\end{equation}

\noindent where $s$ is the dimension of $K_2$, which by Lemma 5.1
is equal to the dimension of $S$, the stoichiometric subspace.

\end{proof}


\bigskip



\begin{thebibliography}{10}



\bibitem{B} ~ Felix E. Browder,{\em On a theorem of Beurling and
Livingston}, \emph{Can.J.Math}, Vol 17(1965), pp. 367-372.

\bibitem{F1}~ Feinberg, M., {\em Chemical reaction network structure and the
stability of complex isothermal reactors I. The deficiency zero
and deficiency one theorems}, {Chemical Engineering Science},
42, 2229-2268, 1987.

\bibitem{F2}~ Feinberg, M., {\em Chemical reaction network structure and the
stability of complex isothermal reactors II. Multiple steady
states for chemical reaction networks of deficiency one},
{Chemical Engineering Science}, 43, 1-25, 1988.

\bibitem{F3}~ Feinberg, M., {\em The existence and uniqueness of steady
states for a class of chemical reaction networks}. {Arch.
Rational Mech. Anal.}, 132, 311-370.

\bibitem{F4}~ Feinberg, M., {\em Multiple steady states for chemical
reaction networks of deficiency one}. {Arch. Rational Mech.
Anal.}, 132, 371-406.

\bibitem{FH1}~ Feinberg, M. \& F.J.M. Horn, 1974, {\em Dynamics of open
chemical systems and the algebraic structure of the underlying
reaction network}, {Chemical Engineering Science}, 29,
775-787.

\bibitem{G}~ Gantmacher, F.R., \em{The theory of matrices}, New York,
Chelsea Pub. Co., 1959.

\bibitem{HJ}~ Horn, F.J.M. \& R.Jackson, 1972, {\em General mass action
kinetics}, {Arch. Rational Mech. Anal.}, 47, 81-116.

\bibitem{J}~  C.Jones, \em{Nonlinear Dynamical Systems: Theory and
Application}, 1996, Lecture note.

\bibitem{N}~ Nachman, A., unpublished notes.





\end{thebibliography}
\end{document}